\documentclass[submission,copyright,creativecommons,a4paper]{eptcs}
\providecommand{\event}{Gandalf 2012} % Name of the event you are submitting to
\usepackage{breakurl}             % Not needed if you use pdflatex only.
\usepackage[latin1]{inputenc} % Any characters can be typed directly from the keyboard, eg Ã©Ã§Ã±
\usepackage{graphicx}  % Add graphics capabilities
\usepackage{amsmath,amsfonts,amssymb,amsthm}
\usepackage{paralist}
\usepackage{stmaryrd}
\usepackage{eucal}

\usepackage{subfigure}

\usepackage{psfrag}
\usepackage{epstopdf}

\newcommand{\Reals}{\mathbb{R}}
   
\newcommand{\Time}{{\sf T}}  

\newcommand{\val}[1]{\mathbf{#1}}
\newcommand{\vals}[1]{\mathit{vals}(#1)}

\newcommand{\ms}[1]{\ifmmode%
\mathord{\mathcode`-="702D\it #1\mathcode`\-="2200}\else%
$\mathord{\mathcode`-="702D\it #1\mathcode`\-="2200}$\fi}

\newcommand{\Tt}{\mathcal{T}}
\newcommand{\M}{\mathcal{M}}
\newcommand{\AutA}{\mathcal{A}}
\newcommand{\AutB}{\mathcal{B}}

\newcommand{\fstate}[1]{#1.\ms{fstate}}
\newcommand{\lstate}[1]{#1.\ms{lstate}}
\newcommand{\arrow}[1]{\mathrel{\stackrel{#1}{\rightarrow}}}

\newcommand{\restr}{\mathop{\lceil}}
\newcommand{\restrrange}{\mathrel{\downarrow}}

\newcommand{\fval}[1]{#1.\ms{fval}}
\newcommand{\lval}[1]{#1.\ms{lval}}
\newcommand{\domain}[1]{{\it dom}(#1)}

\newcommand{\concat}{\mathbin{^{\frown}}}
\newcommand{\hexec}[1]{\ms{execs}(#1)}

\newcommand{\htrace}[1]{{\it trace}(#1)}
\newcommand{\htraces}[1]{{\it traces}(#1)}

\newcommand{\trace}[1]{{\it trace}(#1)}

\newcommand{\ltime}[1]{#1.\ms{ltime}}

\newcommand{\deq}{\mathrel{\stackrel{\scriptscriptstyle\Delta}{=}}}
\newcommand{\hfrag}[1]{\ms{frags}_{#1}}

\newcommand{\type}[1]{\ms{type{(#1)}}}
\newcommand{\dtype}[1]{\ms{dtype{(#1)}}}

\newcommand{\as}[1]{#1_a}

\newcommand{\simrel}{\mathrel{R}}

\newcommand{\ws}[1]{#1_{w}}

\newcommand{\tran}[1]{\stackrel{#1}{\longrightarrow}}

\newcommand{\padact}{\varepsilon}

\newcommand{\padex}{\gamma}

\newtheorem{theorem}{Theorem}
\newtheorem{lemma}{Lemma}
\newtheorem{corollary}{Corollary}
\newtheorem{definition}{Definition}
\newtheorem{proposition}{Proposition}
\newtheorem{example}{Example}
\title{Modelling Implicit Communication in Multi-Agent Systems with Hybrid Input/Output Automata.\thanks{This research has been supported by EC-Project C4C ({\em
Control for Coordination of Distributed Systems}) funded by the European Commission in the 7th EC framework program (Challenge ICT-2007.3.7).}}% and Coordination and Support Action EuRoSurge ({\em European Robotic Surgery}) funded by the European Commission in the 7th EC framework program (FP7-ICT-2011-7).}}
\author{Marta Capiluppi \quad Roberto Segala
\institute{Università di Verona \\
					Dipartimento di Informatica \\
					Verona, Italy}
\email{marta.capiluppi@univr.it, roberto.segala@univr.it}
}
\def\titlerunning{Modelling Implicit Communication in Multi-Agent Systems with Hybrid Input/Output Automata}
\def\authorrunning{M. Capiluppi and R. Segala}

\begin{document}

\maketitle

\begin{abstract}                          % Abstract of not more than 250 words.
We propose an extension of Hybrid I/O Automata (HIOAs) to model agent systems and their implicit communication through perturbation of the environment, like localization of objects or radio signals diffusion and detection. To this end we decided to specialize some variables of the HIOAs whose values are functions both of time and space. We call them {\em world variables}.  Basically they are treated similarly to the other variables of HIOAs, but they have the function of representing the interaction of each automaton with the surrounding environment, hence they can be output, input or internal variables. Since these special variables have the role of simulating implicit communication, their dynamics are specified both in time and space, because they model the perturbations induced by the agent to the environment, and the perturbations of the environment as perceived by the agent. Parallel composition of world variables is slightly different from parallel composition of the other variables, since their signals are summed. The theory is illustrated through a simple example of agents systems.

%The new object, called World Automaton (WA), is built in such a way to preserve as much as possible of the compositionality properties of HIOAs and its underlying theory. From the formal point of view we enrich classical HIOAs with a set of {\em world variables} whose values are functions both of time and space. World variables are treated similarly to local variables of HIOAs, except in parallel composition, where the perturbations produced by world variables are summed. In such a way we obtain a structure able to model both agents and environments, thus inducing a hierarchy in the model and leading to the introduction of a new operator. Indeed this operator, called inplacement, is needed to represent the possibility of an object (WA) of living inside another object/environment (WA). This formalism can be used for modeling and verification of many different scenarios of coordinating agents such as: cooperative search problems (UAVs, submarines, cars); collision avoidance algorithms for autonomous vehicles; signals detection and diffusion (e.g. wireless networks).
\end{abstract}

\section{Introduction}

Many modern complex systems represent agents interacting to achieve a common goal, but reacting in an independent way to external stimuli, following an autonomous decision policy and coordinating using communication. When and where communication fails, the agents need to {\em feel} the environment reacting to its stimuli. 
This is the case, for example, of agents performing a {\em search} mission, such as UAVs \cite{polycarpou} or autonomous underwater vehicles \cite{sousa}, but also of road traffic problems \cite{road,roadcontrol} and autonomous straddle carriers in harbours \cite{adhs12}. These multi-agents problems have been case studies of the European Project CON4COORD (EU FP7 223844) and have motivated the modeling formalism presented in this paper. Indeed what is common to each case study is the presence of a collection of agents that communicate and coordinate to achieve a common goal. Moreover the agents move within an environment that changes dynamically and detect each other's presence not necessarily via direct communication but rather by observations of the environmental changes.
%Following literature, these case studies have been represented by hybrid systems \cite{alur,alur2}, where continuous and discrete dynamics interact. In particular 

We focus on automata-based representations of hybrid systems \cite{alur,alur2}, adding features to a model, in order to keep as much as possible of the underlying theory. Since the motivating case studies need to satisfy some compositionality properties, we choose to start from Hybrid I/O Automata (HIOAs) of \cite{segala}, for which strong results on compositionality exist. 
We add features to represent faithfully situations where a hybrid automaton exists within an environment and derives information about other automata by observing the environment itself, rather than by using any form of direct communication. We will call the exchange of information through observation implicit communication.
%The new automaton is called World Automaton (WA), with the double aim of stressing both the capability of representation of this framework (the world itself) and of representing the reality in the more natural way possible, without adding artificial machineries. For this reason, the basic difference with HIOAs is the introduction of special variables, called {\em world variables}, that take values both in space and time, i.e. their values are indeed functions of the time and space, as occurs in diffusion equations. The introduction of this new kind of variables is motivated again by the case studies: 
Indeed groups of agents usually need to know the environment where they live and move to collect and elaborate data and react in a coordinated way. To do this, they need to exchange stimuli with the surrounding environment, by observation and using sensors. 
Usually the communication between agents acting in a certain area is achieved using artificial machineries, such as supervisors or broadcasting signals. Our aim is to avoid any kind of artificial machinery to model communication between agents and their interaction with the environment, using a more {\em natural} method, based on the human perception, i.e. through observation of the changes in the surroundings of each agent. Moreover direct communication is not always possible, since signals are subject to noise and environmental hostilities, or sometimes it is not the best policy, because sending signals means being intercepted, not considering faults and failures in senders and receivers. Autonomy in making decisions and exchanging information based on the sensing of the reality can be used as a redundant and a faster way of communication. 
To achieve implicit communication, we extend the HIOAs by specializing some variables, called {\em world variables}. They take values both in space and time, i.e. their values are indeed functions of time and space, as occurs in diffusion equations and they represent all the information exchanged between the environment and the agents. World variables represent maps of the changes in the environment as perceived by the agents: at each point in the space and each instant of time their values show the situation that can be sensed by agents standing in that area. To this end, world variables are partitioned into input and output variables: input world variables represent the observations made by the agents sensing the environment, output world variables represent stimuli given by the agents to the environment.
To keep the theory consistent with HIOAs, all the results on semantics are preserved.
%To represent this situation  %In this way the cars send to the outside world a stimulus representing their color, while they receive the information about the color of the surrounding area as input. \\
Moreover we introduce parallel composition using rules similar to the ones for HIOAs, i.e., automata are synchronized on common actions and shared variables, except for output world variables, whose stimuli are summed because their effect on the environment is common.

At the best of our knowledge, there are only a couple of approaches to the presented problem. One has been introduced in \cite{shift} where dynamic networks of hybrid automata are studied. The introduced programming language focuses on dynamical interfaces. Another method has been presented in \cite{cif} where a compositional interchange format (CIF) defined in terms of an interchange automaton is used as a common language to describe objects from the different models for hybrid systems existing in literature. None of these two languages is based on the idea of implicit communication coded by world variables. Our approach is a starting point to solve the problem of dynamical interfaces in a simpler way than the ones proposed. Nevertheless at the current status of our work the presented approach does not solve this problem, even though we started from it. 
We choose to extend HIOAs because of the underlying compositionality theory and because of the input/output distinction of the variables, which we keep in our description. Many other representations of hybrid automata could be used as basis and extended similarly looking at the main theoretical results they have been introduced for. As an example the cited hybrid automata in \cite{alur2} are more focused in reachability issues, but have been studied also for decidability in \cite{decidable}. As stated in \cite{segala} Hybrid Automata (HA) presented in \cite{alur2} are similar to HIOAs in their combined treatment of discrete and continuous activity, but their theory does not address system decomposition issues such as external behavior, implementation relationships and composition. These issues have been addressed in \cite{alur3} by using hybrid reactive modules, but they still differ from the way they are faced by HIOAs, because reactive modules still communicate via shared variables, not via shared actions.  Summarizing, the choice of the HIOA model has been motivated by the fact that their application is more suited for the kind of agent systems and scenarios under study. Indeed the communication via explicit actions, similarly to discrete event automata, is used to model signals communication, while the possibility to trace an external behavior catches the interaction with the environment, which is the basic aim of extending the original formalism with world variables. 

The paper is organized as follows: in Section \ref{sec:wa} we introduce the modeling framework; in Section \ref{sec:semantics} we show and recall the main results on semantics of the proposed model; in Section \ref{sec:parallel} parallel composition of the presented automata is described, showing the main results on composability. The theory is illustrated throughout the paper with a simple example, the interested reader can find a more complex and realistic application in \cite{adhs12}. All the results presented in the paper use the notation of \cite{segala}.

\section{HIOAs with world variables}\label{sec:wa}

\begin{example}\label{ex:description}

Consider a sandy area where two cars move, as in Fig. \ref{fig:example}. We subsume an underline metric space $\Reals^2$. When a car takes a certain position, its pressure provokes a depression of the ground (fig. \ref{fig:lateral}). Hence the car changes the characteristics of the environment in a permanent way, since sand retains the shape. The other car is, then, able to {\em see} where a car has moved (fig. \ref{fig:shape}).
\begin{figure}[htbp]
\begin{center}
\subfigure[A sandy area with cars. \label{fig:example}]{
\includegraphics[width=0.3\textwidth]{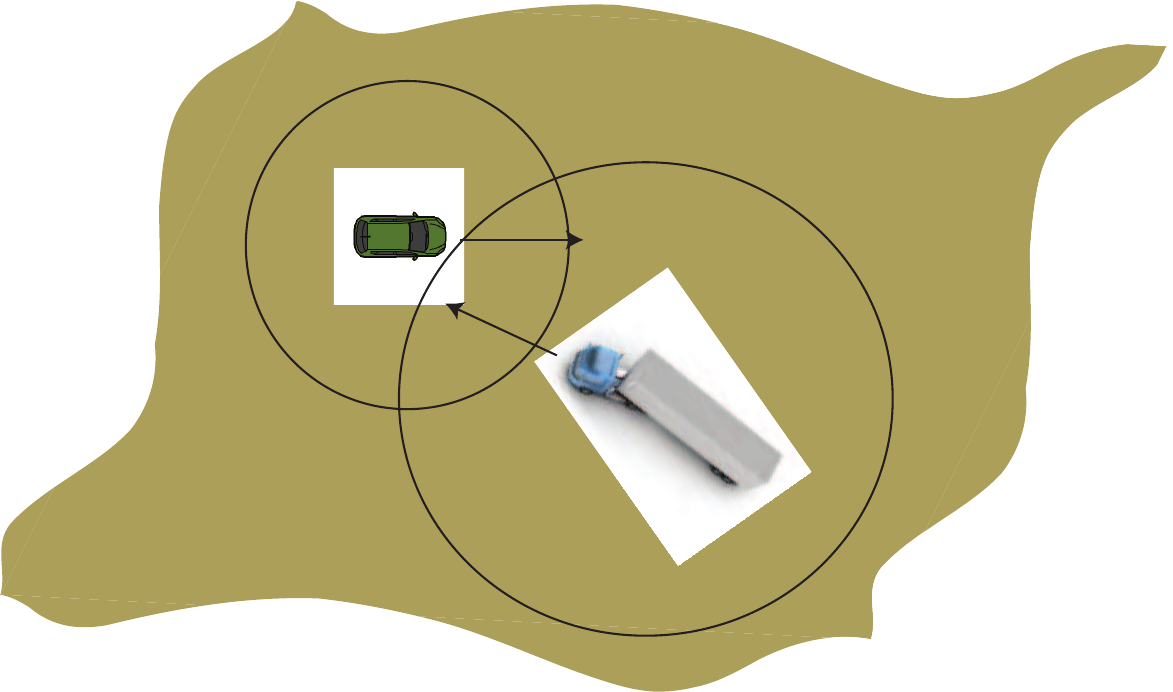}}\hspace{5mm}
\subfigure[Level of the ground. \label{fig:lateral}]{
\includegraphics[width=0.3\textwidth]{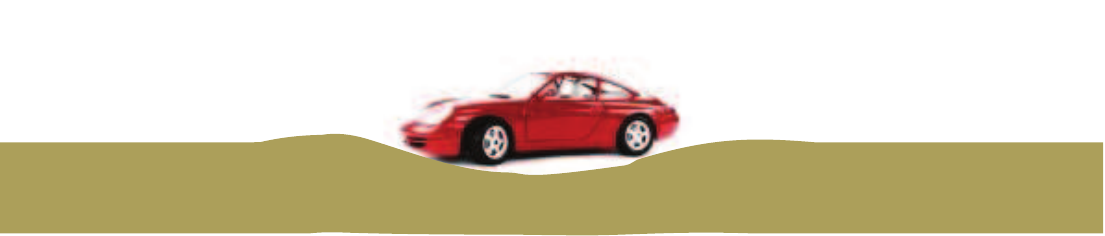}}\hspace{5mm}
\subfigure[A car moving from a point to another in the sandy area. \label{fig:shape}]{
\includegraphics[width=0.3\textwidth]{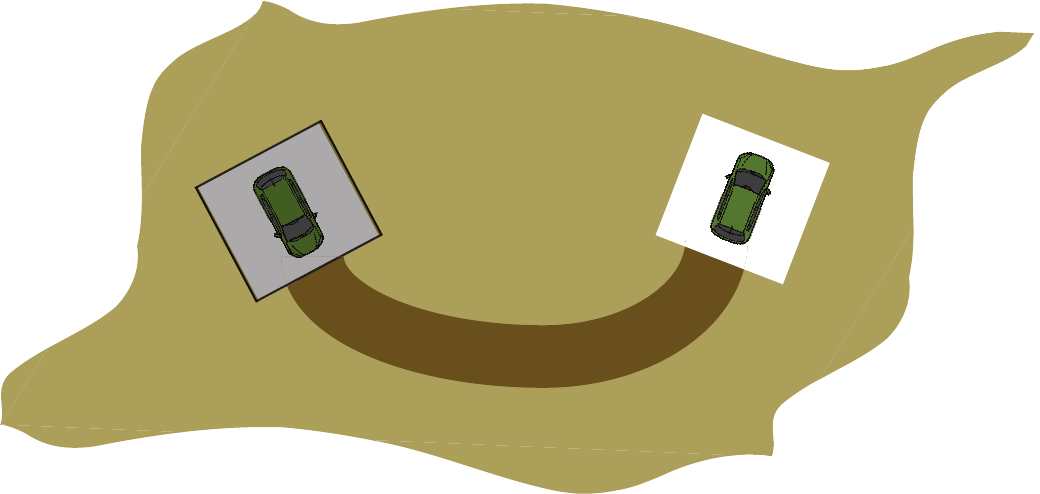}}
\caption{Characteristics of the scenario}
\end{center}
\end{figure}
We aim at avoiding collisions between the two cars. This might be done by equipping each car with some tools to send signals to the other vehicle when approaching, or adding to the system a supervisor knowing at each instant of time the position of both cars. We will call this kind of communication {\em explicit}. Another solution would be to think each car as an {\em intelligent} agent that {\em senses} the surrounding environment and is able to {\em understand} if the other car is too near. We call this kind of communication {\em implicit}. In other words each vehicle should use its sensors to catch the changes in its neighborhood and to calculate the possibility of another car to be in collision risk. The implicit communication is more {\em natural} to us, it does not need artificial machinery, it can be used even in case of hostile environments, where explicit communication is difficult or even impossible, but also when there is need to communicate without sending data through a network. Moreover implicit communication can be used as a redundant mean of communication, when the tools involved in explicit communication fail.
\end{example}

The scenario described in Example \ref{ex:description} is a typical problem of coordination of agents, even though simplified to enlighten only the main challenges the designer has to face in finding a suitable model for this situation. As stated in the Introduction, we decided to use the well known framework of Hybrid I/O Automata (HIOAs) of \cite{segala} to keep the underlying composability theory, very useful in multi-agent problems. 
%HIOAs are extended with the addition of {\em world variables} and their {\em level function} to model in the most {\bf natural} way the interaction between agents and environment, in terms of signals and sensors. We define hereafter a WA recalling also the description of a HIOA (in bold the novelties w.r.t. definition xxx in \cite{segala}). 
%Assume an underlying topological space $\M$. For simplicity the reader may think of $\M$ as a metric space, e.g. $\Reals^3$.
\begin{definition}\label{def:hioa}
{\em Hybrid I/O Automaton (HIOA) \cite{segala}}\\
A HIOA $\AutA$ is a tuple $((U,X,Y),(I,H,O),Q,\Theta,D,\mathcal{T})$ where
\begin{itemize}
%{\bf \item $(\ws{U},\ws{X},\ws{Y})$ are disjoint sets of \emph{world input, inner, and output variables}, respectively. Let $W$ denote the set $\ws{U}\cup\ws{X}\cup\ws{Y}$ of \emph{world variables}.}
\item $(U,X,Y)$ are disjoint sets of \emph{input, internal, and output variables}, respectively. 
%Let $U,X,Y$ denote the sets $\ws{U}\cup\as{U},\ws{X}\cup\as{X},\ws{Y}\cup\as{Y}$ of \emph{input, inner, and output variables}, respectively, and 
Let $V$ denote the set $U\cup X\cup Y$ of \emph{variables}.
\item $(I,H,O)$ are disjoint sets of \emph{input, hidden, and output actions}, respectively. Let $A$ denote the set $I\cup H\cup O$ of \emph{actions}.
\item $Q\subseteq \vals{X}$ is the set of \emph{states}.
\item $\Theta\subseteq Q$ is a nonempty set of \emph{initial states}.
\item $D\subseteq \vals{X}\times A \times \vals{X}$ is the \emph{discrete transition relation}.
\item $\Tt$ is a set of trajectories on
$V$ that satisfy the following axioms
\begin{description}
\item[T1] {\em (Prefix closure)} 
  For every $\tau\in\Tt$ and every $\tau' \leq \tau$, $\tau' \in \Tt$.
\item[T2] {\em (Suffix closure)} 
  For every $\tau \in \Tt$ and every $t\in\domain{\tau}$, $\tau\unrhd t \in \Tt$.
\item[T3] {\em (Concatenation closure)} 
  Let $\tau_0 , \tau_1 , \tau_2, \ldots$ be a sequence of trajectories
  in $\Tt$ such that, for each nonfinal index $i$, $\tau_i$ is closed
  and $\lstate{\tau_i} = \fstate{\tau_{i+1}}$.
  Then $\tau_0 \concat \tau_1 \concat \tau_2 \cdots \in \Tt$.
\end{description}
%{\bf \item $l:S\to\mathbb{N}$ is the {\em level function} extracting the level of a variable or action in any set $S$.}
\end{itemize}
\end{definition}
%Notice that if we group the input, inner and output variables in $U,X,Y$ we obtain an object that coincides with a HIOA (\ref{def:hioa}) in section \ref{sec:math}. Indeed the main differences with HIOAs is given by the fact that $U,X,Y$ sets are partitioned in two different kinds of variables: automaton variables and world variables. 

{\bf Notation}: For each variable $v$, we
assume both a {\em (static) type\/}, \type{v}, which gives the set of values it
may take on, and a {\em dynamic type}, \dtype{v}, which gives the
set of trajectories it may follow. 
A {\em valuation\/} $\val{v}$ for a set of variables $V$ is a function
that associates with each variable $v \in V$ a value in $\type{v}$. We write $\vals{V}$ for the set of valuations for $V$. 
Let $J$ be a left-closed interval of $\Time$ (the time axis) with left endpoint equal to
$0$. Then a {\em $J$-trajectory\/} for $V$ is a function $\tau:
J\rightarrow\vals{V}$,
such that for each $v \in V$, $\tau\restrrange v \in\dtype{v}$.
A {\em trajectory\/} for $V$ is a $J$-trajectory for $V$, for any
$J$. Trajectory $\tau$ is a {\em prefix\/} of trajectory $\tau'$, denoted
by $\tau \leq \tau'$, if $\tau$ can be obtained by restricting
$\tau'$ to a subset of its domain. We define $\tau \unrhd t  \deq  (\tau \restr [t,\infty)) - t$. The concatenation $ \concat$ of two trajectories is obtained
by taking the union of the first trajectory and the function obtained
by shifting the domain of the second trajectory until the start time
agrees with the limit time of the first trajectory; the last valuation
of the first trajectory, which may not be the same as the first valuation
of the second trajectory, is the one that appears in the concatenation.
Prefix, suffix and concatenation operations return trajectories. We define $\fval{\tau}$, the {\em first valuation\/} of $\tau$, to
be $\tau(0)$,
and if $\tau$ is closed ($J$ is a closed interval), we define $\lval{\tau}$, the {\em last
valuation\/} of $\tau$, to be $\tau(\ltime{\tau})$. Given a trajectory $\tau\in\Tt$ we denote $\fval{\tau}\restr X$ by
$\fstate{\tau}$ and, if $\tau$ is closed,
we denote $\lval{\tau}\restr X$ by $\lstate{\tau}$.
We write $f \restr P$ for the restriction of function $f$ to set $P$,
that is, the function $g$ with $\domain{g} = \domain{f} \cap P$ such that
$g(c) = f(c)$ for each $c \in \domain{g}$. 
If $f$ is a function whose range is a set of functions and $P$ is a set,
then we write $f \restrrange P$ for the function $g$ with $\domain{g}
= \domain{f}$ such that $g(c) = f(c) \restr P$ for each $c \in \domain{g}$. 
For more detail the interested reader can refer to \cite{segala}.

The reader can notice that the main difference with respect to the model introduced by \cite{alur, alur2} is that locations are not explicit, indeed they are given by state variables, trajectories and transitions. Moreover transitions from one state to another do not occur by crossing guards or leaving invariants, but they occur because of actions arising (see executions definition in Section \ref{sec:semantics}).

\begin{example}\label{ex:hioa}
Imagine now to describe the scenario in example \ref{ex:description} using hybrid automata. To represent HIOAs we use a variant of the TIOA language \cite{TIOA}, with some extensions for hybrid systems \cite{helicopter}. The HIOA of a car is reported in fig. \ref{fig:car1}. 
\begin{figure}[htbp]
{\bf type} Rad = $\Reals| 2\pi$\\
{\bf hioa} Car\\
{\bf variables}

\qquad      {\bf input}  collisionrisk: Bool, groundlevel: Bool

\qquad	{\bf internal} $\phi$: Rad, $p_T$: Real$^2$, $m$: Real, $vel$:Real

\qquad      {\bf output} $P$: Real$^2$, $K$: Real

{\bf trajectories}

%\qquad   $\xi(t,p)=\left\{\begin{array}{ll}
%\mbox{black} & \mbox{if } p\in f(\phi,p_T) \\
%\mbox{white} & \mbox{otherwise}
%\frac{gm}{p_T.1}-\frac{s(t,p)}{p_T.1} & \mathrm{if }\ p\leq p_T \land \neg \mathrm{Fail}\\
%0 & \mathrm{otherwise}
%\end{array}\right.$;

\qquad   $K(t)=z(m, f(\phi,p_T))$;

\qquad $\ms{vel}(t)=\left\{\begin{array}{ll}
			0 & \mbox{if } \ms{collisionrisk}\\
			0.5 & \mbox{if } \ms{groundlevel}\\
			1 & \mbox{otherwise}.
			\end{array}\right.$;
			
\qquad  $P(t)=p_T(t)$.
%\myammer{inserire colore e livello}

\caption{HIOA representing a car.}
\label{fig:car1}
%\end{minipage}
%\end{scriptsize}
\end{figure}
It has an output variable $K$ representing the ground pressure provided by the car and an output variable $P$ representing the car position. The input variables are: the level of the ground {\em groundlevel} as a boolean saying if the level is low (1) or high (0); the {\em collisionrisk} saying if another car is in collision risk (1) or not (0). A function $f$ is defined, giving the surface of the ground occupied by the car area starting from its position $p_T$ and its orientation angle $\phi$. We can imagine that $f$ returns a rectangle centered in $p_T$ with orientation $\phi$. The pressure variable is updated with a function $z$ depending on the mass $m$ and area of the car. The velocity $\ms{vel}$ of the car is 0 if collisionrisk is true. Similarly the car slows down when groundlevel is true.
For the sake of simplicity we used a boolean variable to represent the ground level changes, but any other function can be used, such as more general and complex diffusion equations.
Note that we need to provide the system with an {\em external supervisor} which, taking as input the position and pressure of each car in the area at each instant of time, calculates the collision risk and the ground level around it. Basically the supervisor needs to know each car direction and position at each instant of time for calculating the possibility of a collision with other cars moving in the same area and the level of the ground along the trajectory the car is following. We do not present the design of such a supervisor because it is out of the scope of this paper.
Note also that this is just a possible representation of the scenario described in example \ref{ex:description}. We used this simple way to show the need of using some external machinery (e.g. a supervisor) to model the interaction of agents with the environment.
\end{example}

In example \ref{ex:hioa} we are not able to represent implicit communication without adding some artificial machinery (in this case we used an external supervisor). Since our aim is to represent the system in a more natural way, we extend HIOA modeling framework to catch this aspect. To do this we specialize some variables of the HIOA, calling them {\em world variables}. The name is due to the fact that we want them to represent the connection between the agents and the surrounding world. Moreover world variables represent the changes in the environment as might be perceived by the agents. Hence the set of variables $V$ is partitioned in a set $W$ of world variables and a set $S$ of standard automaton variables. The set $W$ is partitioned in sets $(\ws{U},\ws{X},\ws{Y})$ of \emph{world input, internal, and output variables}, respectively, such that:
$\ws{U}\subseteq U, \ws{X}\subseteq X, \ws{Y}\subseteq Y$. To avoid confusion, we will add to automaton variables the subscript $a$: $\as{U}, \as{X}, \as{Y}$.

The main difference between world and automaton variables is that the type of world variables is a function of time and space, not only of time as in standard automaton variables. Hence world variables values (and trajectories) will depend both on the instant of time and the position in the underlying space. Formally, if we assume an underlying topological space $\M$, $\val{w}: (\mathcal{T}\times\M)\to B$ for every $w\in W$, where $\mathcal{T}$ is the time axes and $B$ is a set. For simplicity the reader may think of $\M$ as a metric space, e.g. $\Reals^3$.
An automaton $\AutA$ will use its world inputs $\ws{U}$ to receive stimuli from the world it lives in. Analogously it will use its world outputs $\ws{Y}$ to give stimuli to the world it lives in. Finally internal world variables $\ws{X}$ are used to represent the world characteristics of $\AutA$. 
%In this description, like in the HIOA theory, the automaton variables $A$ are used by the automaton to directly synchronize with other automata living in the same world. 
To keep the theory consistent with previous descriptions of automata, all the $X$ variables represent persistent characteristics of the system. We will call HIOAs with world variables HIOAWs.
\begin{example}\label{ex:hioaw}
We now represent the car in fig. \ref{fig:car1} with a HIOAW, extending the TIOA language to include world variables. Note that world variables are always described using their trajectories in time and space, i.e. they are described for any instant of time $t$ and any point in space $p$. 
Each car is represented by a HIOAW as in fig. \ref{fig:car}. It has an output world variable $k$ representing the ground pressure provided by the car and an output world variable representing the car color $\xi$. The input world variables are: the level of the ground $g$ and its color $c$. Each car perceives the ground level through a boolean variable $g$ saying if the ground is low (1) or high (0). We used the boolean representation for the sake of simplicity. Of course any other function, like diffusion equations, may be used. Each car can check the color of the ground at each point of the area by the variable $c$, which represents a kind of colored map of the area.
%We assume that the two cars cannot be in the same place at the same time, because of safety control policies, but if accident are taken into account, this possibility can be included in the description. A function $f$ is defined for cars, giving the surface of the ground occupied by the car area starting from its position $p_T$ and its orientation angle $\phi$. We can imagine that $f$ returns a rectangle centered in $p_T$ with orientation $\phi$. 
We assume that the color variable $\xi$ takes the value black for all the points inside the car area given by $f$ and white outside. The pressure variable $k$ is updated with a function $h$ depending on the mass $m$ and area of the car, associating to each point in the area of the car the value of its pressure in time, and to each point outside the area of the car a 0 value. 
Two actions {\em collision, level} represent the possibility that another car is in the neighborhood and that the level of the ground in the neighborhood is low, respectively. Action {\em collision} activates a boolean variable {\em stop} if there is any black point $p^*$ in the neighborhood of the car, which is calculated by the function $q$ returning a circle of radius $r$ (bigger than the semi-diagonal of the rectangle representing the area of the car) and centered in $p_T$, but excluding the area of the car given by function $f$. Action {\em level} activates a boolean variable {\em slow} if there is any point $p^*$ in the neighborhood of the car for which the ground level variable $g$ is true, i.e. the level of the ground is low.
Hence the velocity $\ms{vel}$ of the car is 0 if {\em stop} is true. Similarly the car slows down when {\em slow} is true. 
All the presented equations describing the car dynamics are very simple, but the description of the motion is out of the scope of this paper. Indeed they can be substituted by any other equations.
The reader can notice that in fig. \ref{fig:car1} the position of the car is explicit in variable $P$, which is an output that must be collected by the supervisor at each instant of time to check where the automaton is in the space. In the HIOAW of fig. \ref{fig:car} the position is embedded in the world variables and does not need to be explicitly put in an automaton variable. Indeed both color and pressure world variables carry the information about the position of the automaton in the space, due to their nature.
\end{example}
\begin{figure}[tb]
{\bf type} Rad = $\Reals| 2\pi$

{\bf hioaw} Car\\
{\bf world variables}

\qquad	{\bf input}  $g$: Bool, $c$: Color;

\qquad	{\bf output} $k$: Real, $\xi$: Color;

{\bf automaton variables}

\qquad	{\bf internal} $\phi$: Rad, $p_T$: Real$^2$, $m$: Real, $vel$:Real, $r$:Real, stop: Bool, slow: Bool;

{\bf actions}

\qquad     {\bf hidden} collision, level;

{\bf transitions}

\qquad {\bf hidden} collision

\qquad {\bf pre} $ \exists p^*\in q(p_T,r, f(\phi,p_T)) \mbox{ s.t. } c(t,p^*)=\mbox{black}$

\qquad {\bf eff} stop $=$ true;

\qquad {\bf hidden} level

\qquad {\bf pre}  $\exists p^*\in q(p_T,r, f(\phi,p_T)) \mbox{ s.t. } g(t,p^*)=$ true

\qquad{\bf eff} slow $=$ true;

{\bf trajectories}

\qquad   $\xi(t,p)=\left\{\begin{array}{ll}
\mbox{black} & \mbox{if } p\in f(\phi,p_T) \\
\mbox{white} & \mbox{otherwise}
%\frac{gm}{p_T.1}-\frac{s(t,p)}{p_T.1} & \mathrm{if }\ p\leq p_T \land \neg \mathrm{Fail}\\
%0 & \mathrm{otherwise}
\end{array}\right.$;

\qquad   $k(t,p)=h(m, f(\phi,p_T))$;

\qquad $\ms{vel}(t)=\left\{\begin{array}{ll}
			0 & \mbox{if stop}\\
			0.5 & \mbox{if slow}\\
			1 & \mbox{otherwise}.
			\end{array}\right.$
%\myammer{inserire colore e livello}

\caption{HIOAW representing a car.}
\label{fig:car}
%\end{minipage}
%\end{scriptsize}
\end{figure}

The reader can notice that the automaton in fig. \ref{fig:car} has some input world variables. Here we considered the environment as an abstract entity, modifying and being modified by the agents living in it. As in the human sensing, the agents moving in an environment can catch these modifications as changes with respect to the nominal conditions of the surrounding area and interact with them. In the same way the agents change the environment. World variables aim at representing this exchange of implicit information because they give a map of environmental changes at each point of the space and each instant of time, without need of artificial machineries such as a supervisor. 

\section{Semantics}\label{sec:semantics}

Executions of HIOAWs are defined as executions of HIOAs: an \emph{execution fragment} of a HIOAW $\AutA$ is an ($A,V$)-sequence $\alpha=\tau_0 a_1 \tau_1 a_2 \tau_2 \ldots$, where $a_i\in A$, $\tau_i\in \Tt$; if $\tau_i$ is not the last trajectory of $\alpha$, then $\lstate{\tau_i} \arrow{a_{i+1}} \fstate{\tau_{i+1}}$. An execution fragment $\alpha$ is defined to be an {\em execution\/} if $\fstate{\alpha}$ is a start state, that is, $\fstate{\alpha} \in \Theta$. Results on executions of HIOAs are valid also for HIOAWs. 

A {\em trace} of an execution fragment $\alpha$ captures the external behavior of a HIOAW, i.e. what it is needed to identify an automaton from outside. Calling $E=I\cup O$, $Z=U\cup Y$, a trace of a HIOAW $\AutA$ is then the ($E,Z$)-restriction of $\alpha$. All the results on traces on HIOAs are still valid and exactly stated for HIOAWs.
%All the results on executions, traces and simulation relations on HIOAs are still valid and exactly stated for WAs. This is possible because we have defined WAs as a well-defined extension of HIOAs. For the definitions the reader can refer to section \ref{sec:prelim} and to \cite{segala}.
%The main difference with the semantics of HIOAs is that for traces it is possible to introduce a taxonomy based on the level of variables and actions we want to be seen outside the WA. 
%As recalled in section \ref{sec:prelim} a {\em trace} of an execution fragment $\alpha$ captures the external behavior of $\AutA$, i.e. what it is needed to identify an automaton from outside. Calling $E=I\cup O$, $Z=U\cup Y$, a trace of a WA $\AutA$ is then the ($E,Z$)-restriction of $\alpha$. We will, then, call trace the ($E,Z$)-restriction of $\alpha$ at all levels, supposing that the external behavior is captured at all levels of the automaton. When needed it is possible to restrict the behavior to a specific level $i$ by defining the ($E[i],Z[i]$)-restriction of $\alpha$, called $[i]$-trace. Moreover it is possible to define restrictions at more than one level, e.g. defining $[1,2]$-traces as the ($E[1,2],Z[1,2]$)-restriction of $\alpha$. Many different combination of levels can be taken into account while defining traces, depending on what it meant to be visible from outside the automaton. 
%All the results here used refer to all level traces, hence, since the above defined taxonomy is a partition of the all level traces, they are also valid for restricted levels traces.
We say that a low-level specification $\AutA$ {\em implements} a high-level specification $\AutB$
if any behavior of $\AutA$ is also an allowed behavior of $\AutB$.
\begin{definition}\label{def:comp_imp}
Automata $\AutA_1$ and $\AutA_2$ are {\em comparable\/} if they have
the same external interface, that is, if world and local input and output sets of variables of $\AutA_1$ are equal to the corresponding sets of $\AutA_2$ and $E_1 = E_2$ at all levels.
%$\glob{Y_1} = \glob{Y_2}$, $\aut{Y_1}=\aut{Y_2}$, $\local{Y_1}=\local{Y_2}$, $\glob{U_1} = \glob{U_2}$, $\aut{U_1}=\aut{U_2}$, $\local{U_1}=\local{U_2}$, $E_1 = E_2$.
If $\AutA_1$ and $\AutA_2$ are comparable then we say that
$\AutA_1$ {\em implements} $\AutA_2$, denoted by $\AutA_1 \leq \AutA_2$, if
%the traces of $\AutA_1$ are included among those of $\AutA_2$, that is, if
$\htraces{\AutA_1} \subseteq \htraces{\AutA_2}$.
\end{definition}

Simulation relations between HIOAWs are defined as for HIOAs in Section 4.3 of \cite{segala}. We report here the definition:

\begin{definition}\label{def:simulation}
Let $\AutA$ and $\AutB$ be comparable automata. A {\em simulation\/} from $\AutA$
to $\AutB$ is a relation $\simrel\ \subseteq Q_\AutA \times Q_\AutB$
satisfying the following conditions, for all states $\val{x}\restr Q_\AutA\triangleq \val{x}_\AutA$ and
$\val{x}\restr Q_\AutB\triangleq \val{x}_\AutB$ of $\AutA$ and $\AutB$, respectively:
\begin{enumerate}
\item If $\val{x}_\AutA \in \Theta_\AutA$ then there exists a
  state $\val{x}_\AutB \in \Theta_\AutB$ such that $\val{x}_\AutA \simrel \val{x}_\AutB$.
\item If $\val{x}_\AutA \simrel \val{x}_\AutB$ and $\alpha$ is an execution fragment
of $\AutA$ consisting of one action surrounded by two point trajectories,
with $\fstate{\alpha} = \val{x}_\AutA$,
then $\AutB$ has a closed execution fragment $\beta$
with $\fstate{\beta} = \val{x}_\AutB$,
$\trace{\beta} = \trace{\alpha}$, and
$\lstate{\alpha} \simrel \lstate{\beta}$.
\item If $\val{x}_\AutA \simrel \val{x}_\AutB$ and $\alpha$ is an execution fragment
of $\AutA$ consisting of a single closed trajectory,
with $\fstate{\alpha} = \val{x}_\AutA$,
then $\AutB$ has a closed execution fragment $\beta$
with $\fstate{\beta} = \val{x}_\AutB$,
$\trace{\beta} = \trace{\alpha}$, and
$\lstate{\alpha} \simrel \lstate{\beta}$.
\end{enumerate}
\end{definition}
 
Results on trace inclusion for simulation of HIOAs are valid also for HIOAWs. We also report here an important corollary on simulation relations which will be used in the rest of the paper.
\begin{corollary}
  \label{cor:ha-sim-ht}
  Let $\AutA$ and $\AutB$ be comparable automata and let $\simrel$ be a
  simulation from $\AutA$ to $\AutB$.
  Then $\htraces{\AutA}\subseteq\htraces{\AutB}$.
\end{corollary}
%For the notation the reader can refer to \cite{segala}.

\subsection{Padding of executions}

We introduce now the notion of {\em padding} of executions that will be used in the following proofs.

\begin{definition}
A {\em padded execution} of a HIOAW $\AutA$ is an $(A\cup\{\padact\},V)-$sequence $\padex=\tau_0 a_1 \tau_1 a_2 \tau_2 a_3\ldots$ such that if $a_i=\padact$ then $\lstate{\tau_{i-1}}=\fstate{\tau_i}$.
%Let $\alpha$ be an execution $\tau_0 a_1 \tau_1 a_2 \tau_2 a_3 \ldots$, where $\forall i, \tau_i\in\trajs{V}, a_i\in A$. A {\em padded execution} $\padex{\alpha}$ of $\alpha$ is an execution defined in $(A\cup\{\padact\},V)$, s.t. $\padex{\alpha}\restr(A,V)=\alpha$.
\end{definition}

\begin{definition}{\em Padding.}\\
We call {\em padding} of an execution $\alpha$ any padded execution obtained by $\alpha$ by extending the actions set with $\padact$.
\end{definition}
For example a padded execution of an execution $\alpha=\tau_0 a_1 \tau_1 a_2 \tau_2 a_3 \ldots$ is $\padex=\tau'_0\padact\tau''_0 a_1 \tau_1 a_2 \tau_2 a_3 \ldots$, where $\tau'_0\concat\tau''_0=\tau_0$. 

\begin{definition}
The restriction of a padded execution $\padex$ to a set of actions $A'$ and a set of variables $V'$ is the $(A',V')$-restriction of $\padex$.
\end{definition}

\begin{lemma}\label{lm:padtoexec}
Let $\padex$ be a padded execution of $\AutA$. Then there exists $\alpha$, execution of $\AutA$, for which $\padex$ is a padding. 
\end{lemma}
\begin{proof}
Let $A,V$ be the sets of actions and variables of $\AutA$, respectively. Then, by definition of restriction of padded executions and by definition of executions, $\alpha=\padex\restr(A,V)$ is an execution of $\AutA$. By definition of padding $\padex$ is a padding of $\alpha$.
\end{proof}

\begin{lemma}\label{lm:padrestr}
Let $\alpha$ be an execution of $\AutA$ defined in $(A,V)$ and $\padex$ a padding of $\alpha$. Let $A'\subseteq A,V'\subseteq V$, then $\alpha\restr(A',V')=\padex\restr(A',V')$.
\end{lemma}
\begin{proof}
Straightforward by definition of restriction of a padded execution and of an execution and by definition of padding.
\end{proof}

\begin{definition}
A trace of a padded execution $\padex$ is defined as $\padex\restr(E,Z)$.
\end{definition}

\begin{lemma}\label{lm:padtrace}
Let $\alpha$ be an execution of $\AutA$ and $\padex$ a padding of $\alpha$, then $\htrace{\alpha}=\htrace{\padex}$.
\end{lemma}
\begin{proof}
Straightforward by lemma \ref{lm:padrestr} and definition of trace of a padded execution.
\end{proof}

\begin{lemma}
Let $\padex$ be a padding of $\alpha$, execution of $\AutA$, and let $\padex'$ be a prefix of $\padex$. Then $\padex'\restr(A,V)$ is a prefix of $\alpha$.
\end{lemma}

\begin{lemma}\label{lm:padbuild}
Given $n$ executions, it is always possible to find $n$ paddings of these executions such that all corresponding trajectories have the same length.
%Let $\alpha=\tau_0 a_1 \tau_1 a_2 \tau_2\ldots,\beta=\upsilon_0 b_1\upsilon_1 b_2 \upsilon_2\ldots$ be executions of $\AutA$. It is always possible to find $n$ padded executions $\padex{\alpha}=\tau'_0 a'_1 \tau'_1 a'_2 \tau'_2\ldots$ such that:
%\begin{itemize}
%\item $\domain{\tau'_0}=\min\{\domain{\tau_0},\domain{\upsilon_0}\}$
%\item if $a_1=b_1$ then $a'_1=a_1$
%\item if $a_1\neq b_1$ then $a'_1=\padact$
%\item ...
%\end{itemize}
\end{lemma}

\section{Parallel composition}\label{sec:parallel}
%In our framework Parallel composition models the interaction and communication of two agents living in the same world, i.e. of two WAs at the same level. To this end we introduce compatibility conditions to prevent undesired interactions between different levels for the WAs that have to be composed. 

In this section we introduce parallel composition for HIOAWs. First of all some compatibility conditions have to be stated.

\begin{definition}\label{def:parcomp}
Two HIOAWs $\AutA_1$ and $\AutA_2$ are compatible if 
\begin{enumerate}
%\item $V_1[1,.]\cap V_2[1,.] = \emptyset$, $A_1[1,.]\cap A_2[1,.] = \emptyset$,
\item $(U_{w1}\cup U_{w2}) \cap (Y_{w1}\cup Y_{w2}) = \emptyset$.
\item $H_1\cap A_2 = H_2\cap A_1 = \emptyset$,
\item $X_1\cap V_2 = X_2\cap V_1 = \emptyset$,
\item $O_1\cap O_2 = \emptyset$,
\item $Y_1\cap Y_2 = \emptyset$.
\end{enumerate}
\end{definition}

The reader may notice that conditions 2 to 5 are the classical compatibility conditions for HIOAs.  The first condition %states that all inner levels (higher than 0) are disjoint, and therefore no communication can occur at such levels. This means that generated worlds are disjoint. The second condition 
states that no explicit communication between the two HIOAWs is possible via world variables. Indeed, by definition, explicit communication between HIOAWs occurs only via automaton I/O variables, whereas world variables are used for implicit communication. 
%Since by the first condition communication may occur only at level 0, all the other properties are interesting only for level 0, even if not specified.
These conditions, when not satisfied by the HIOAWs, can be obtained by changing variables names. 

\begin{definition}\label{def:parallel}
{\em Parallel composition}\\
If $\AutA_1,\AutA_2$ are two compatible HIOAWs, then their composition $\AutA_1\|\AutA_2$ is defined as the structure $\AutA$ where:
\begin{enumerate}
\item $\ws{U} = U_{w1}\cup U_{w2}$, $\ws{X} = X_{w1}\cup X_{w2}$, $\ws{Y} = Y_{w1}\cup Y_{w2}$
\item $\as{Y}=Y_{a1}\cup Y_{a2}$, $\as{X}=X_{a1}\cup X_{a2}$, $\as{U}=(U_{a1}\cup U_{a2})\setminus\as{Y}$
\item $O=O_1\cup O_2$, $I=(I_1\cup I_2)\setminus O$ and $H=H_1\cup H_2$
\item $Q=\{\val{x}\in \vals{X}\mid\val{x}\restr X_1\in Q_1\land \val{x}\restr X_2\in Q_2\}$
\item $\Theta=\{\val{x}\in Q \mid \val{x}\restr X_1\in \Theta_1\land \val{x}\restr X_2\in \Theta_2\}$
\item $D = \{(\val{x},a,\val{x}') \mid $ for each $i\in\{1,2\}$ either $a\in A_i$ and $\val{x}\restr X_i \tran{a} \val{x}^{\prime}\restr X_i $, or $a\notin A_i$ and $\val{x}\restr X_i = \val{x}^{\prime}\restr X_i\}$.
\item $\Tt=\{\tau\mid$ there exists $\tau_1\in\Tt_1,\tau_2\in\Tt_2$ such that 
$\tau\restrrange (V_i\setminus (Y_{w1}\cap Y_{w2}))=\tau_i\restrrange (V_i\setminus  (Y_{w1}\cap Y_{w2})), i\in\{1,2\}$ and
$\tau\restrrange  (Y_{w1}\cap Y_{w2})=\tau_1\restrrange  (Y_{w1}\cap Y_{w2})+\tau_2\restrrange (Y_{w1}\cap Y_{w2})\}$
\end{enumerate}
%and $l(v)=l_1(v)$ if $v\in V_1$, $l(v)=l_2(v)$ if $v\in V_2$.
\end{definition}
%{\bf Notation}: in the above definition $\tau,\tau_i$ indicate trajectories, the symbol $\restrrange$ indicates the projection as defined in Section 2.1 of \cite{segala}.

This definition of parallel composition is very similar to the one for HIOAs. The only two differences are given by the first and the last conditions.
The first condition depends on compatibility: there is no communication between the two HIOAWs via world variables.
The last condition indeed is the main difference with HIOAs composition. 
It might be that the two composing automata have some output world variables with the same kind of information for the external world. These output world variables will have the same name and then their intersection is not empty.
For those variables it is necessary to sum the trajectories as defined in the following. We call sum any generic operator with the same characteristics of the sum in $\Reals$. In the following we will define the {\em sum} as an additive operator in a group. 
Let $\tau_0,\tau_1$ be two trajectories with the same time domain, such that:
$\tau_0:[0,t]\to(V_0\to\mathcal{D})$ and
$\tau_1:[0,t]\to(V_1\to\mathcal{D})$,
where $V_0,V_1$ are sets of variables. Let $\mathcal{D}$ be a domain of values for variables in $V_0,V_1$ (e.g. $\Reals$), such that its structure is a (commutative) group $\mathcal{G}$, with an operator $+_G$ and an identity element called $0_G$. Note that the subscript $G$ will be omitted when it is clear from the context.
\begin{definition}\label{def:trajsum}
The sum of $\tau_0,\tau_1$ is defined as:
\[
(\tau_0+\tau_1)(t)(v)=\left\{
\begin{array}{ll}
\tau_i(t)(v) & \mbox{if } v\in V_i\setminus V_{1-i}\\
\tau_0(t)(v)+\tau_1(t)(v) & \mbox{if } v\in V_0\cap V_1
\end{array}\right.
\]
\end{definition}
For the sake of simplicity in the following we will consider the operation of sum in $\Reals$. But all the results presented in this paper are still valid using any other operator with the same mathematical characteristics.

We report here some lemmas on trajectories that will be used in the following proof.
\begin{lemma}\label{lm:restrprojtraj}
Let $\tau$ be a trajectory in $V$. Let $I\subseteq\domain{\tau}$ and $V^\prime\subseteq V$. Then $(\tau\restr I)\restrrange V'=(\tau\restrrange V')\restr I$.
\end{lemma}
\begin{lemma}\label{lm:restrsuffixtraj}
Let $\tau$ be a trajectory in $V$. Let $V^\prime\subseteq V$. Then $(\tau \unrhd t)\restrrange V'=(\tau\restrrange V') \unrhd t$.
\end{lemma}
\begin{lemma}\label{lm:restrconcattraj}
Let $\tau$ be a trajectory in $V$ such that $\tau=\tau_0\concat\tau_1\concat\tau_2\concat\ldots$. Let $V^\prime\subseteq V$. Then $(\tau_0\concat\tau_1\concat\tau_2\concat\ldots)\restrrange V'=(\tau_0\restrrange V')\concat(\tau_1\restrrange V')\concat(\tau_2\restrrange V')\concat\ldots$.
\end{lemma}

\begin{proposition}\label{th:autparallel}
The composition of two HIOAWs is a HIOAW.
\end{proposition}

\begin{proof}
We show that $\AutA_1\|\AutA_2$ satisfies the properties of a HIOAW. Disjointness of the $U,X,Y$ sets follows from disjointness of the same sets in $\AutA_1$ and $\AutA_2$ and compatibility. Similarly for the actions. Nonemptiness of starting state follows from nonemptiness of starting states of $\AutA_1$ and $\AutA_2$ and disjointness of $X_1$ and $X_2$.
We verify the {\bf T} properties of trajectories (see definition \ref{def:hioa}). Let $C_{12}$ be $Y_{w1}\cap Y_{w2}$.
 \begin{description}
 
  \item[T1]  We want to prove that for every $\tau\in\Tt$ and every $\tau' \leq \tau$, $\tau' \in \Tt$. Let $\tau$ be a trajectory in $\Tt$. Let $i\in\{1,2\}$. By the definition of parallel composition there exists $\tau_1\in \Tt_1,\tau_2\in\Tt_2$ such that $\tau\restrrange (V_i\setminus C_{12})=\tau_i\restrrange (V_i\setminus C_{12})$, and $\tau\restrrange  C_{12}=\tau_1\restrrange C_{12}+\tau_2\restrrange C_{12}.$
Let $\tau^\prime\leq\tau$. By definition of prefix we have that $\tau^\prime=\tau\restr I$ with $I=\domain{\tau^\prime}\subseteq\domain{\tau}$. Hence we can state that $\tau^\prime\restrrange(V_i\setminus C_{12})=(\tau\restr I)\restrrange (V_i\setminus C_{12})$.
By lemma \ref{lm:restrprojtraj}
$(\tau\restr I)\restrrange (V_i\setminus C_{12})=(\tau\restrrange (V_i\setminus C_{12}))\restr I$.
By definition of parallel composition and again by lemma \ref{lm:restrprojtraj}
$(\tau\restrrange (V_i\setminus C_{12}))\restr I=(\tau_i\restrrange (V_i\setminus C_{12}))\restr I=(\tau_i\restr I)\restrrange (V_i\setminus C_{12})$.
Let $\tau^\prime_1=\tau_1\restr I$ and $\tau^\prime_2=\tau_2\restr I$, then 
$(\tau_i\restr I)\restrrange (V_i\setminus C_{12})=\tau^\prime_i\restrrange (V_i\setminus C_{12})$.
Analogously, for the second statement of parallel composition of trajectories we have that
$\tau^\prime\restrrange C_{12}=(\tau\restr I)\restrrange C_{12}=(\tau\restrrange C_{12})\restr I=(\tau_1\restrrange C_{12})\restr I+(\tau_2\restrrange C_{12})\restr I=(\tau_1\restr I)\restrrange C_{12}+(\tau_2\restr I)\restrrange C_{12}=\tau^\prime_1\restrrange C_{12}+\tau^\prime_2\restrrange C_{12}$.
Hence $\tau'\in\Tt$.

  \item[T2] We want to prove that for every $\tau \in \Tt$ and every $t\in\domain{\tau}$, $\tau\unrhd t \in \Tt$. Let $\tau$ be a trajectory in $\Tt$. Let $i\in\{1,2\}$. By the definition of parallel composition there exists $\tau_1\in \Tt_1,\tau_2\in\Tt_2$ such that
  $\tau\restrrange (V_i\setminus C_{12})=\tau_i\restrrange (V_i\setminus C_{12}),$ and
$\tau\restrrange  C_{12}=\tau_1\restrrange C_{12}+\tau_2\restrrange C_{12}.$
   Hence, since $\domain{\tau_1}=\domain{\tau_2}=\domain{\tau}$, by lemma \ref{lm:restrsuffixtraj} we have that 
$ (\tau\unrhd t)\restrrange (V_i-C_{12})=(\tau\restrrange (V_i-C_{12}))\unrhd t=(\tau_i\restrrange (V_i-C_{12}))\unrhd t= (\tau_i\unrhd t)\restrrange (V_i-C_{12}) $.
  Moreover 
$ (\tau\unrhd t)\restrrange C_{12}=(\tau\restrrange C_{12})\unrhd t=(\tau_1\restrrange C_{12})\unrhd t+(\tau_2\restrrange C_{12})\unrhd t= (\tau_1\unrhd t)\restrrange C_{12}+(\tau_2\unrhd t)\restrrange C_{12}$.
  Recall that by the properties of trajectories $\tau\unrhd t$ is still a trajectory. Hence $\tau\unrhd t \in \Tt$.

  \item[T3] We want to prove that set $\Tt$ is closed under concatenation. Let $\tau_0,\tau_1,\tau_2,\ldots$ be a sequence of trajectories in $\Tt$, such that, for each nonfinal index $j$, $\tau_j$ is closed and  $\lstate{\tau_j}=\fstate{\tau_{j+1}}$. Let $\tau$ be $\tau_0\concat\tau_1\concat\tau_2\concat\ldots$. Let $i\in\{1,2\}$. 
 By definition of parallel composition for each $\tau_j$, $\exists \tau_{1j}, \tau_{2j}$ such that
  $\tau_j\restrrange (V_i\setminus C_{12})=\tau_{ij}\restrrange (V_i\setminus C_{12}),$ and
$\tau_j\restrrange  C_{12}=\tau_{1j}\restrrange C_{12}+\tau_{2j}\restrrange C_{12}.$
Let $\tau_i$ be $\tau_{i0}\concat\tau_{i1}\concat\tau_{i2}\concat\ldots$. Hence by lemma \ref{lm:restrconcattraj}
$\tau\restrrange(V_i\setminus C_{12})=(\tau_0\restrrange (V_i\setminus C_{12}))\concat(\tau_1\restrrange (V_i\setminus C_{12}))\concat(\tau_2\restrrange (V_i\setminus C_{12}))\concat\ldots=(\tau_{i0}\restrrange (V_i\setminus C_{12}))\concat(\tau_{i1}\restrrange (V_i\setminus C_{12}))\concat(\tau_{i2}\restrrange (V_i\setminus C_{12}))\concat\ldots=(\tau_{i0}\concat\tau_{i1}\concat\tau_{i2}\concat\ldots)\restrrange (V_i\setminus C_{12})=\tau_i\restrrange (V_i\setminus C_{12})$.
Moreover
$\tau\restrrange  C_{12}=(\tau_0\restrrange  C_{12})\concat(\tau_1\restrrange  C_{12})\concat(\tau_2\restrrange  C_{12})\concat\ldots
=(\tau_{10}\restrrange C_{12}+\tau_{20}\restrrange C_{12})\concat(\tau_{11}\restrrange C_{12}+\tau_{21}\restrrange C_{12})\concat(\tau_{12}\restrrange C_{12}+\tau_{22}\restrrange C_{12})\concat\ldots=
((\tau_{10}\restrrange C_{12})\concat(\tau_{11}\restrrange C_{12})\concat(\tau_{12}\restrrange C_{12})\concat\ldots)+((\tau_{20}\restrrange C_{12})\concat(\tau_{21}\restrrange C_{12})\concat(\tau_{22}\restrrange C_{12})\concat\ldots)=(\tau_{10}\concat\tau_{11}\concat\tau_{12}\concat\ldots)\restrrange C_{12} + (\tau_{20}\concat\tau_{21}\concat\tau_{22}\concat\ldots)\restrrange C_{12}=\tau_1\restrrange C_{12}+\tau_2\restrrange C_{12}$.
Hence $\tau \in \Tt$. 
 \end{description}
\end{proof}

\begin{example}
Consider again example \ref{ex:description}. Suppose to have a HIOAW representing a car as in fig. \ref{fig:car} in the sandy area, called $\AutB_1$. Another car represented by $\AutB_2$ (again of type represented in fig. \ref{fig:car}) enters the sandy area. We want to compose the two cars. Variables and actions of $\AutB_1$ are labelled by the subscript 1, the ones of $\AutB_2$ by the subscript 2. The obtained HIOAW $\AutB_1\|\AutB_2$ is represented in fig. \ref{fig:par_cars}. Notice that the effect of the output world variables are summed: the ground pressure $k$ of $\AutB_1\|\AutB_2$ represents a sort of a map of the values taken by the pressures given by the cars in the considered area, the same for the color $\xi$. Indeed here we did not make any constraints of two cars being at the same point at a time, because the model can take into account also collisions.
\end{example}

We now show that simulation relation and trace inclusion are preserved by composition. The main difficulty compared to the analogous results in \cite{segala} is that output world variables sum their effects. This means that it is not possible anymore to project executions of a composite system to obtain executions of the components. Rather we have to show that for each execution of the composite system there are executions of the components that can be pasted together.
\begin{figure}[htbp]
%\begin{scriptsize}
{\bf hioaw} $\AutB_1\|\AutB_2$\\
{\bf world variables}

\qquad	{\bf input}  $g$: Bool, $c$: Color;

\qquad	{\bf output} $k$: Real, $\xi$: Color;

{\bf automaton variables}

\quad	{\bf internal} $\phi_1$: Rad, $p_{T1}$: Real$^2$, $m_1$: Real, $vel_1$:Real, $r_1$:Real,
$\phi_2$: Rad, $p_{T2}$: Real$^2$, $m_2$: Real, $vel_2$:Real, $r_2$:Real, stop$_1$: Bool, stop$_2$: Bool, slow$_1$: Bool, slow$_2$: Bool;

{\bf actions}

\qquad  {\bf hidden} collision$_1$, collision$_2$, level$_1$, level$_2$;

{\bf transitions}

\qquad {\bf hidden} collision$_1$

\qquad {\bf pre} $ \exists p^*\in q(p_{T1},r_1) \mbox{ s.t. } c(t,p^*)=\mbox{black}$

\qquad {\bf eff} stop$_1=$ true;

\qquad {\bf hidden} collision$_2$

\qquad {\bf pre} $ \exists p^*\in q(p_{T2},r_2) \mbox{ s.t. } c(t,p^*)=\mbox{black}$

\qquad {\bf eff} stop$_2=$ true;

\qquad {\bf hidden} level$_1$

\qquad {\bf pre}  $\exists p^*\in q(p_{T1},r_1) \mbox{ s.t. } g(t,p^*)=$ true

\qquad{\bf eff} slow$_1=$ true;

\qquad {\bf hidden} level$_2$

\qquad {\bf pre}  $\exists p^*\in q(p_{T2},r_2) \mbox{ s.t. } g(t,p^*)=$ true

\qquad{\bf eff} slow$_2=$ true;

{\bf trajectories}

\qquad   $\xi(t,p)=\left\{\begin{array}{ll}
\mbox{black} & \mbox{if } p\in f(\phi_1,p_{T1})\lor p\in f(\phi_2,p_{T2}) \\
\mbox{white} & \mbox{otherwise}
\end{array}\right.$;

\qquad   $k(t,p)=h(m_1, f(\phi_1,p_{T1}))+h(m_2, f(\phi_2,p_{T2}))$;

\qquad $\ms{vel_1}(t)=\left\{\begin{array}{ll}
			0 & \mbox{if stop}_1\\
			0.5 & \mbox{if slow}_1 \\
			1 & \mbox{otherwise}.
			\end{array}\right.$
			
\qquad $\ms{vel_2}(t)=\left\{\begin{array}{ll}
			0 & \mbox{if stop}_2 \\
			0.5 & \mbox{if slow}_2 \\
			1 & \mbox{otherwise}.
			\end{array}\right.$

\caption{HIOAW representing parallel composition of $\AutB_1$ and $\AutB_2$.}
\label{fig:par_cars}
%\end{scriptsize}
\end{figure}

%We now introduce a projection lemma, that relates executions of the composed automaton to executions of the component automata.
\begin{lemma}\label{lm:compexec}
Let $\AutA=\AutA_1\parallel\AutA_2$ and let $\alpha$ be an execution fragment of $\AutA$. Then $\exists \alpha_1,\alpha_2$ execution fragments of $\AutA_1$ and $\AutA_2$ respectively, such that 
\begin{enumerate}
\item $\alpha\restr(A_i,V_i\setminus C_{12})=\alpha_i\restr(A_i,V_i\setminus C_{12}), i=1,2$, and 
\item $\alpha\restr(\emptyset, C_{12})=\alpha_1\restr(\emptyset, C_{12})+\alpha_2\restr(\emptyset, C_{12})$,
\end{enumerate} 
with $C_{12}=(Y_{w1}\cap Y_{w2})$.
\end{lemma}

\begin{proof}
Let $\alpha=\tau_0 a_1\tau_1 a_2\tau_2 a_3\ldots \in\hfrag{\AutA}$. 
By definition of parallel composition, since each $\tau_j\in\Tt$ there exists $\tau_{j1}\in\Tt_1,\tau_{j2}\in\Tt_2$ such that: $\tau_j\restrrange (V_1\setminus  C_{12})=\tau_{j1}\restrrange (V_1\setminus  C_{12})$, $\tau_j\restrrange (V_2\setminus  C_{12})=\tau_{j2}\restrrange (V_2\setminus  C_{12})$ and $\tau_j\restrrange  C_{12}=\tau_{j1}\restrrange  C_{12}+\tau_{j2}\restrrange C_{12}.$
    Hence by definition of padding we can build two padded executions $\padex_1=\tau_{01} a'_1 \tau_{11} a'_2 \tau_{21} a'_3\ldots, \padex_2=\tau_{02} a''_1 \tau_{12} a''_2 \tau_{22} a''_3\ldots$ of $\AutA_1$ and $\AutA_2$ respectively, where
\[
a'_j=\left\{\begin{array}{ll}
a_j & \mbox{ if } a_j\in A_1\\
\padact & \mbox{ otherwise}
\end{array}\right.
\qquad
a''_j=\left\{\begin{array}{ll}
a_j & \mbox{ if } a_j\in A_2\\
\padact & \mbox{ otherwise}
\end{array}\right.
\]
%By lemma (projection of paddings) we have that 
%\[
%\alpha\restr(A_1,V_1\setminus C_{12})=\padex{\alpha'}\restr(A_1,V_1\setminus C_{12})
%\] 
Let $i\in\{1,2\}$. By construction of $\padex_i$ it is
$\alpha\restr(A_i,V_i\setminus C_{12})=\padex_i\restr(A_i,V_i\setminus C_{12})$.
By lemma \ref{lm:padtoexec} it is possible to define the execution $\alpha_i$ of $\AutA_i$ for which $\padex_i$ is a padding. Then by lemma \ref{lm:padrestr} it holds:
$\padex_i\restr(A_i,V_i\setminus C_{12})=\alpha_i\restr(A_i,V_i\setminus C_{12})$
which proves point 1 of this lemma. %Similarly for $\alpha_2$. 
Moreover, since the projection of an execution on an empty set of action gives a trajectory, we have that
$\alpha\restr(\emptyset, C_{12})=\padex_1\restr(\emptyset, C_{12})+\padex_2\restr(\emptyset, C_{12})=\alpha_1\restr(\emptyset, C_{12})+\alpha_2\restr(\emptyset, C_{12})$.
By definition \ref{def:trajsum} the last statement proves point 2 of this lemma.
\end{proof}

The following lemma from HIOAs applies directly to HIOAWs. The proof is reported in \cite{segala}. 
%says that we obtain the same result for an
%execution fragment $\alpha$ of a composition if we first extract the
%trace and then restrict to one of the components, or if we first
%restrict to the component and then take the trace.
\begin{lemma} \label{lemma-comm-proj}
  Let $\AutA=\AutA_1\| \AutA_2$, and let $\alpha$ be an execution
  fragment of $\AutA$.  Then, for $i=1,2$,
  $\trace{\alpha}\restr(E_i,Z_i) = \trace{\alpha\restr(A_i,V_i)}$.
\end{lemma}
The following proposition relates the set of traces of a composite
automaton to the sets of traces of the component automata.
%It expresses the idea that a trace of a
%composition ``projects'' to obtain traces of the components and that traces of
%components can be ``pasted together'' to obtain a trace of the composition.
\begin{proposition}
  \label{lemma-htraces-proj}
Let $\AutA=\AutA_1\| \AutA_2$ and $\beta$ a trace of $\AutA$. Then $\exists \beta_1,\beta_2$ traces of $\AutA_1,\AutA_2$ respectively, such that 
\begin{enumerate}
\item $\beta\restr(E_i,Z_i\setminus C_{12})=\beta_i\restr(E_i,Z_i\setminus C_{12}), i=1,2$ and 
\item $\beta\restr(\emptyset, C_{12})=\beta_1\restr(\emptyset, C_{12})+\beta_2\restr(\emptyset, C_{12})$,
\end{enumerate} 
with $C_{12}=(Y_{w1}\cap Y_{w2})$.
\end{proposition}

\begin{proof}
Let $\beta$ be a trace of $\AutA$. By definition of trace $\exists\alpha\in\hexec{\AutA}$ such that $\beta=\htrace{\alpha}$. 
By Lemma \ref{lm:compexec}, $\exists \alpha_1,\alpha_2$ execution fragments of $\AutA_1$ and $\AutA_2$ respectively, such that $\alpha\restr(A_i,V_i\setminus C_{12})=\alpha_i\restr(A_i,V_i\setminus C_{12}), i=1,2$, and $\alpha\restr(\emptyset, C_{12})=\alpha_1\restr(\emptyset, C_{12})+\alpha_2\restr(\emptyset, C_{12})$. 
Let $\beta_1=\htrace{\alpha_1}$ and $\beta_2=\htrace{\alpha_2}$. We want to prove that $\beta\restr(E_i,Z_i\setminus C_{12})=\beta_i\restr(E_i,Z_i\setminus C_{12}), i=1,2$ and $\beta\restr(\emptyset, C_{12})=\beta_1\restr(\emptyset, C_{12})+\beta_2\restr(\emptyset, C_{12})$.
By Lemma \ref{lemma-comm-proj}, $\beta\restr(E_i,Z_i) = \trace{\alpha\restr(A_i,V_i)}$. Moreover by the properties of projection of executions, we have that $\beta\restr(E_i,Z_i\setminus C_{12}) = \trace{\alpha\restr(A_i,V_i\setminus C_{12})}$ and $\beta\restr(\emptyset,C_{12}) = \trace{\alpha\restr(\emptyset,C_{12})}$. 
Furthermore by the properties of projection of executions we have that $(\alpha_i\restr(A_i,V_i\setminus C_{12}))\restr(E_i,Z_i\setminus C_{12})=\alpha_i\restr(A_i\cap E_i,(V_i\cap Z_i)\setminus C_{12})$. Since by definition $A_i\cap E_i=E_i$ and $V_i\cap Z_i=Z_i$, we obtain $\alpha_i\restr(A_i\cap E_i,(V_i\cap Z_i)\setminus C_{12})=\htrace{\alpha}\restr(E_i,Z_i\setminus C_{12})$. Similarly for projections on $C_{12}$. 
\end{proof}
%The proof of this theorem follows the one in \cite{segala} and it is not reported here for lack of space.
%
%Finally we present a theorem introducing substitutivity property:

The next two theorems prove the results on substitutivity for implementation and simulation relations.

\begin{theorem}
\label{th:sub_imp}
Let $\AutA_1$ and $\AutA_2$ be comparable HIOAWs with $\AutA_1 \leq \AutA_2$.
Let $\AutB$ be a HIOAW compatible with each of $\AutA_1$ and $\AutA_2$. Then $\AutA_1 \| \AutB$ and $\AutA_2 \| \AutB$ are comparable and $\AutA_1 \| \AutB \leq \AutA_2 \| \AutB$.
\end{theorem}

\begin{proof}
Let $\alpha$ be an execution of $\AutA_1\parallel\AutB$. By lemma \ref{lm:compexec}, two executions $\alpha_1, \alpha_B$ exist, such that $\alpha_1\in\hexec{\AutA_1}$, $\alpha_B\in\hexec{\AutB}$ and:
$\alpha\restr(A_1,V_1\setminus C_{1B})=\alpha_1\restr(A_1,V_1\setminus C_{1B})$,
$\alpha\restr(A_B,V_B\setminus C_{1B})=\alpha_B\restr(A_B,V_B\setminus C_{1B})$,
$\alpha\restr(\emptyset, C_{1B})=\alpha_1\restr(\emptyset, C_{1B})+\alpha_B\restr(\emptyset, C_{1B})$,
with $C_{1B}=Y_{w1}\cap Y_{wB}$.
By lemma \ref{lm:padbuild} we can take paddings of $\alpha, \alpha_1, \alpha_B$ such that the $j^{th}$ trajectory has the same length for all $j$. Let these paddings be $\padex, \padex_1, \padex_B$ respectively with $\padex=\tau_0 a_1 \tau_1 a_2 \tau_2 a_3 \ldots$, $\padex_1=\tau_{01} a'_{1} \tau_{11} a'_{2} \tau_{21} a'_{3} \ldots$ and $\padex_B=\tau_{0B} a''_{1} \tau_{1B} a''_{2} \tau_{2B} a''_{3} \ldots$.
Since $\AutA_1 \leq \AutA_2$ and by compatibility, we can find an execution $\alpha_2$ of $\AutA_2$ with the same trace of $\alpha_1$ and a padding of $\alpha_2$ following lemma \ref{lm:padbuild}. We write $\padex_2=\tau_{02} a'''_{1} \tau_{12} a'''_{2} \tau_{22} a'''_{3} \ldots$. 
By the definition of composition the execution of $\AutA_2\parallel\AutB$ obtained by $\padex_2$ and $\padex_B$ will be $\padex'=\tau'_0 b_1 \tau'_1 b_2 \tau'_2 b_3 \ldots$, where  
$\tau'_j\restrrange (V_2\setminus C_{2B})=\tau_{j2}\restrrange (V_2\setminus  C_{2B})$,
$\tau'_j\restrrange (V_B\setminus C_{2B})=\tau_{jB}\restrrange (V_B\setminus  C_{2B})$,
$\tau'_j\restrrange  C_{2B}=\tau_{j2}\restrrange  C_{2B}+\tau_{jB}\restrrange C_{2B}$, 
where $C_{2B}=Y_{w2}\cap Y_{wB}$. This is valid even if the trajectories in the padded executions have not the length of the original trajectories, by definition of prefix of a trajectory and prefix closure of trajectories in a HIOAW.
Actions $b_i$ might be different, but by construction, compatibility and lemma \ref{lm:padtrace} we have that $\padex'$ has the same trace of $\padex$ hence of $\alpha$. Indeed the (padded) executions can differ only in their internal variables (state), but they do not influence the traces (external variables). For this reason we can state that $\htraces{\AutA_1\parallel\AutB}\subseteq\htraces{\AutA_2\parallel\AutB}$, hence, by definition \ref{def:comp_imp} of implementation, $\AutA_1 \| \AutB \leq \AutA_2 \| \AutB$.
\end{proof}

\begin{corollary}\label{cor:sim2impl}
Let $\AutA_1$ and $\AutA_2$ be compatible HIOAWs, and let $\simrel$ be a simulation relation between $\AutA_1$ and $\AutA_2$. Let $\AutB$ be a HIOAW compatible with each of $\AutA_1$ and $\AutA_2$. Then $\AutA_1 \| \AutB \leq \AutA_2 \| \AutB$.
\end{corollary}

\begin{proof}
Since $\AutA_1 \simrel \AutA_2$, by corollary \ref{cor:ha-sim-ht}, $\htraces{\AutA_1}\subseteq\htraces{\AutA_2}$. By definition \ref{def:comp_imp} of implementation, $\AutA_1 \leq \AutA_2$. By theorem \ref{th:sub_imp} this implies that for any $\AutB$, $\AutA_1 \| \AutB \leq \AutA_2 \| \AutB$.
\end{proof}

\begin{theorem}
\label{th:sub_sim}
Let $\AutA_1$ and $\AutA_2$ be compatible HIOAWs, and let $\simrel$ be a simulation relation between $\AutA_1$ and $\AutA_2$.
Let $\AutB$ be a HIOAW compatible with each of $\AutA_1$ and $\AutA_2$. Then $\exists \simrel'$ such that $(\AutA_1 \| \AutB) \simrel' (\AutA_2 \| \AutB)$.
\end{theorem}

\begin{proof}
Let $\simrel'$ be a relation between $\AutA_1 \| \AutB$ and $\AutA_2 \| \AutB$ such that for each $x_1\in  Q_1$, $x_2\in Q_2$, $x_B,x'_B \in Q_B$,
$(x_1,x_B)\simrel'(x_2,x'_B) \mbox{ iff } (x_1 \simrel x_2) \land (x_B=x'_B)$.
We prove that $\simrel'$ is a simulation relation by proving that $\simrel'$ satisfies each point of definition \ref{def:simulation}.
\begin{enumerate}
\item Since for each $x_1\in  \Theta_1$, $x_2\in \Theta_2$, $x_1 \simrel x_2$, then by definition of $\simrel'$, for each initial state $(x_1,x_B)$ of $\AutA_1\|\AutB$ and each initial state $(x_2,x_B)$ of $\AutA_2\|\AutB$, $(x_1,x_B)\simrel'(x_2,x_B)$, with $x_B \in \Theta_B$.
\item Let $\alpha$ be an execution fragment of $\AutA_1\|\AutB$ consisting of one action surrounded by two point trajectories, with $\fstate{\alpha}=(x_1,x_B)$. Let $\lstate{\alpha}=(x'_1,x'_B)$. Since $x_1\simrel x_2$ then $\exists x'_2\in Q_2$ such that $x'_1\simrel x'_2$. Since $\AutA_1\simrel\AutA_2$, by corollary \ref{cor:sim2impl}, $\AutA_1 \| \AutB \leq \AutA_2 \| \AutB$. Then there exists $\beta$ execution fragment of $\AutA_2\|\AutB$ with the same trace of $\alpha$ and $\fstate{\beta}=(x_2,x_B)$. By definition of parallel composition there exists an action bringing the state to $(x'_2,x'_B)$. Hence by definition of $\simrel'$ we have that $(x'_1,x'_B)\simrel'(x'_2,x'_B)$.
\item Let $\alpha$ be and execution fragment of $\AutA_1\|\AutB$ such that $\alpha=\tau\in\Tt$ closed and with $\fstate{\alpha}=(x_1,x_B)$.  Let $\beta$ be an execution fragment of $\AutA_2\|\AutB$ such that $\fstate{\beta}=(x_2,x_B)$. Let $\lstate{\alpha}=(x'_1,x'_B)$ and $\lstate{\beta}=(x'_2,x'_B)$. Since $x_1\simrel x_2$, by definition of $\simrel'$ it is $(x_1,x_B)\simrel'(x_2,x_B)$. By corollary \ref{cor:sim2impl}, there exists an execution fragment $\beta$ of $\AutA_2\|\AutB$ with the same trace of $\alpha$.
\end{enumerate}
\end{proof}

\section{Conclusions}

In this paper we have proposed an extension of the Hybrid I/O Automaton model of \cite{segala} to provide a natural representation of the fact that objects move in a world that they can observe and modify. We started from the analysis of the case studies of the C4C project, representing agents that move in a dynamical environment and have to achieve a goal by coordination.
Besides the classical signals that automata send to each other either via discrete communication events or shared continuous variables, we specialized some variables of HIOAs to let them communicate implicitly by affecting their surrounding world and observing the effects on the worlds of the activity of other automata. This mechanism for interaction turns out to be adequate for compositional analysis, which is one of the main features of HIOAs that we wanted to keep in an extended model. Indeed we introduced the notion of parallel composition, and proved compositionality results. 
The natural extension of this formalism to model environment has been reported in \cite{tec_rep}, leading to a hierarchical representation of automata and introducing the ability of composing them {\em vertically} into nested worlds. We presented in this paper a toy example to show the application of the theory, but a more complex and reality-based application can be found in \cite{adhs12}. The simulation tools are under study.
Future research directions include the ability to describe scenarios where automata are created and destroyed and where communication links change dynamically.

\bibliographystyle{eptcs}
\bibliography{gandalf12}

\end{document}